\documentclass[runningheads]{llncs}

\let\llncssubparagraph\subparagraph
\let\subparagraph\paragraph
\usepackage[compact]{titlesec}
\let\subparagraph\llncssubparagraph

\usepackage[compact]{titlesec}
\usepackage[nocompress]{cite}
\usepackage{pdfpages}
\usepackage{tabularx,booktabs}
\newcolumntype{Y}{>{\centering\arraybackslash}X}
\usepackage{float}
\usepackage{amsmath}
\usepackage{amsfonts}
\usepackage{amssymb}
\usepackage{graphicx}
\usepackage{calc}
\usepackage[normalem]{ulem}
\usepackage{float}
\usepackage[utf8]{inputenc}
\usepackage{pgfplots}
\usepackage[misc]{ifsym}
\usepackage{makecell}
\usepackage{multirow}
\usepackage{subfig}
\DeclareMathOperator*{\argmax}{argmax}

\usepackage{algorithm}
\usepackage{algpseudocode}
\makeatletter
\renewcommand{\ALG@name}{Heuristic}
\makeatother

\begin{document}
	\title{Edge User Allocation with Dynamic Quality of Service\thanks{This manuscript has been accepted for publication at the 17th International Conference on Service-Oriented Computing and may be published in the book series Lecture Notes in Computer Science. All copyrights reserved to Springer Nature Switzerland AG, Gewerbestrasse 11, 6330 Cham, Switzerland.}}
	\titlerunning{}
	
	\author{Phu Lai\inst{1} \and Qiang He\inst{1}$^{\textrm{(\Letter)}}$ \and Guangming Cui\inst{1} \and Xiaoyu Xia\inst{2} \and Mohamed Abdelrazek\inst{2} \and Feifei Chen\inst{2} \and John Hosking\inst{4} \and John Grundy\inst{3} \and Yun Yang\inst{1}}
	
	\authorrunning{P. Lai et al.}
	\institute{
		Swinburne University of Technology, Hawthorn, Australia \\
		\email{\{tlai,qhe,gcui,yyang\}@swin.edu.au} \and
		Deakin University, Burwood, Australia \\
		\email{\{xiaoyu.xia,mohamed.abdelrazek,feifei.chen\}@deakin.edu.au} \and
		Monash University, Clayton, Australia \\
		\email{john.grundy@monash.edu} \and
		The University of Auckland, Auckland, New Zealand \\
		\email{j.hosking@auckland.ac.nz}
	}
	
	\maketitle
	
	\begin{abstract}
		In edge computing, edge servers are placed in close proximity to end-users. App vendors can deploy their services on edge servers to reduce network latency experienced by their app users. The edge user allocation (EUA) problem challenges service providers with the objective to maximize the number of allocated app users with hired computing resources on edge servers while ensuring their fixed quality of service (QoS), e.g., the amount of computing resources allocated to an app user. In this paper, we take a step forward to consider dynamic QoS levels for app users, which generalizes but further complicates the EUA problem, turning it into a dynamic QoS EUA problem. This enables flexible levels of quality of experience (QoE) for app users. We propose an optimal approach for finding a solution that maximizes app users' overall QoE. We also propose a heuristic approach for quickly finding sub-optimal solutions to large-scale instances of the dynamic QoS EUA problem. Experiments are conducted on a real-world dataset to demonstrate the effectiveness and efficiency of our approaches against a baseline approach and the state of the art.
		
	\end{abstract}
	
	\section{Introduction} \label{sec:introduction}
	Mobile and Internet-of-Things (IoT) devices, including mobile phones, wearables, sensors, etc., have become extremely popular in modern society \cite{cerwall2018ericsson}. The rapid growth of those devices have increased the variety and sophistication of software applications and services such as facial recognition \cite{soyata2012cloud}, interactive gaming \cite{chen2015decentralized}, real-time, large-scale warehouse management \cite{ding2008application}, etc. Those applications usually require intensive processing power and high energy consumption. Due to the limited computing capabilities and battery power of mobile and IoT devices, a lot of computing tasks are offloaded to app vendors' servers in the cloud. However, as the number of connected devices is skyrocketing with the continuously increasing network traffic and computational workloads, app vendors are facing the challenge of maintaining a low-latency connection to their users. 
	
	Edge computing -- sometimes often referred to as \emph{fog computing} -- has been introduced to address the latency issue that often occurs in the cloud computing environment \cite{bonomi2012fog}. A usual edge computing deployment scenario involves numerous edge servers deployed in a distributed manner, normally near cellular base stations \cite{hu2015mobile}. This network architecture significantly reduces end-to-end latency thanks to the close proximity of edge servers to end-users. The coverage areas of nearby edge servers usually partially overlap to avoid non-serviceable areas -- the areas in which users cannot offload tasks to any edge server. A user located in the overlapping area can connect to one of the edge servers covering them (\emph{proximity constraint}) that has sufficient computing resources (\emph{resource constraint}) such as CPU, storage, bandwidth, or memory. Compared to a cloud data-center server, a typical edge server has very limited computing resources, hence the need for an effective and efficient resource allocation strategy.
	
	\begin{figure}
		\centering
		\includegraphics[page=1,scale=0.32]{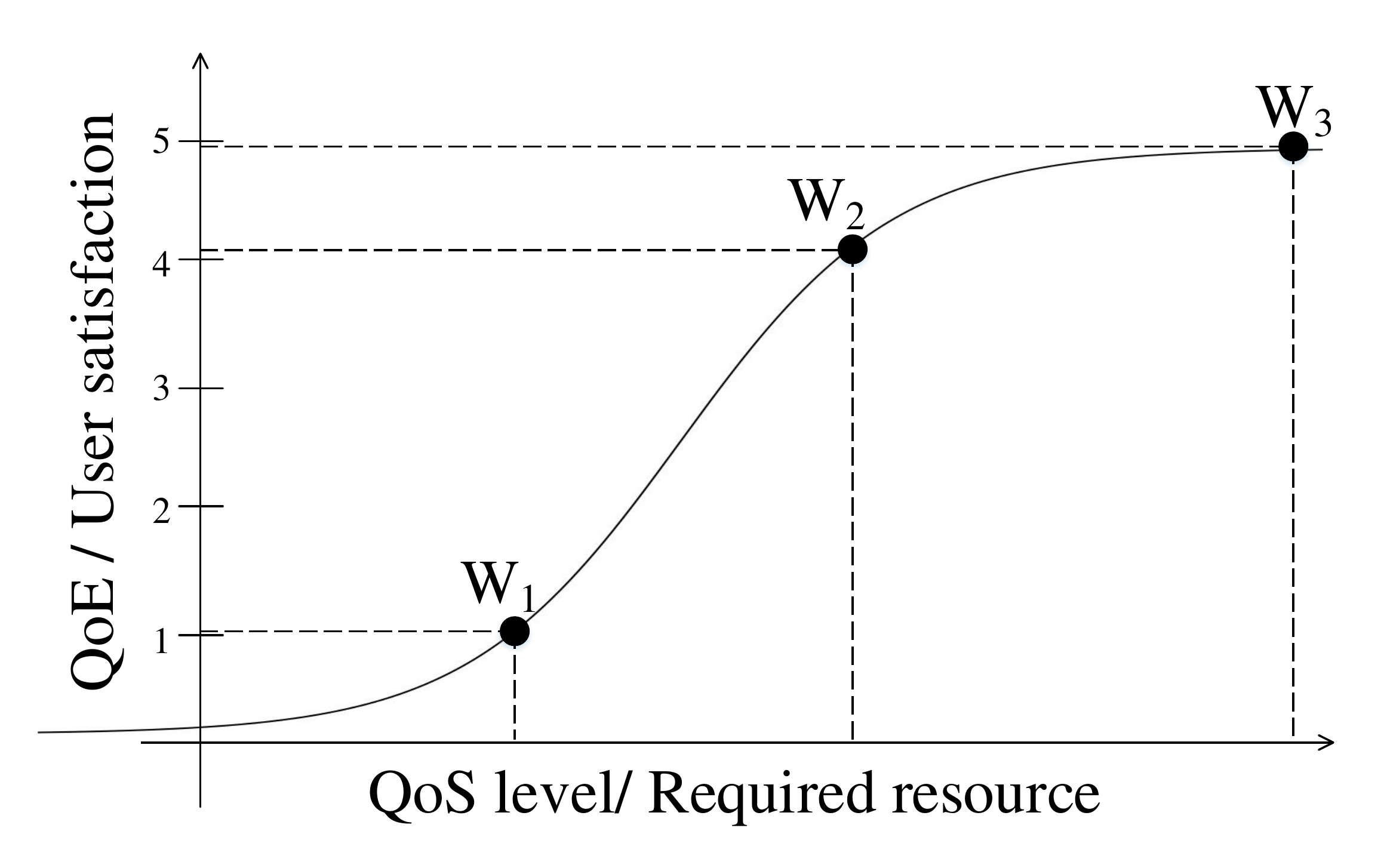}
		\caption{Quality of Experience - Quality of Service correlation}
		\label{fig:QoE_QoS}
	\end{figure}
	
	Naturally, edge computing is immensely dynamic and heterogeneous. Users using the same service have various computing needs and thus require different levels of quality of service (QoS), or computational requirements, ranging from low to high. Tasks with high complexity, e.g. high-definition graphic rendering, eventually consume more computing resources in an edge server. A user's satisfaction, or quality of experience (QoE), varies along with different levels of QoS. Many researchers have found that there is a quantitative correlation between QoS and QoE, as visualized in Fig. \ref{fig:QoE_QoS} \cite{fiedler2010generic, alreshoodi2013survey, hossfeld2011quantification}. At one point, e.g. $ W_3 $, the user satisfaction tends to converge so that the QoE remains virtually unchanged at the highest level regardless of how high the QoS level is. 
	
	Consider a typical game streaming service for example, gaming video frames are rendered on the game vendor's servers then streamed to player's devices. For the majority of players, there is no perceptible difference between 1080p and 1440p video resolution on a mobile device, or even between 1080p and UHD from a distance farther than $1.5$x the screen height regardless of the screen size \cite{lachat2015perception}. Servicing a 1440p or UHD video certainly consumes more resources (bandwidth, processing power), which might be unnecessary since most players are likely to be satisfied with 1080p in those cases. Instead, those resources can be utilized to serve players who are currently unhappy with the service, e.g. those experiencing poor 240p or 360p graphic, or those not able to play at all due to all nearby servers being overloaded. Therefore, the app vendor can lower the QoS requirements of high demanding users, potentially without any remarkable downgrade in their QoE, in order to better service users experiencing low QoS levels. This way, app vendors can maximize users' overall satisfaction measured by their overall QoE. In this context, our research aims at allocating app users to edge servers so that their overall QoE is maximized.
	
	We refer to the above problem as a \textit{dynamic QoS edge user allocation} (EUA) problem. Despite being critical in edge computing, this problem has not been extensively studied. Our main contributions are as follows:
	\begin{itemize}
		\item We define and model the dynamic QoS EUA problem, and prove its $\mathcal{NP}$-hardness.
		\item We propose an optimal approach based on integer linear programming (ILP) for solving the dynamic QoS EUA and develop a heuristic approach for finding sub-optimal solutions to large-scale instances of the problem efficiently.
		\item Extensive evaluations based on a real-world dataset are carried out to demonstrate the effectiveness and efficiency of our approaches against a baseline approach and the state of the art.
	\end{itemize}
	
	The remainder of the paper is organized as follows. Section \ref{sec:motivating_example} provides a motivating example for this research. Section \ref{sec:problem_definition} defines the dynamic QoS problem and proves that it is $\mathcal{NP}$-hard. We then propose an optimal approach based on ILP and an efficient sub-optimal heuristic approach in Sect. \ref{sec:approaches}. Section \ref{sec:evaluation} evaluates the proposed approaches. Section \ref{sec:related_work} reviews the related work. Finally, we conclude the paper in Sect. \ref{sec:conclusion}.
	\newline
	
	\section{Motivating Example}\label{sec:motivating_example}
	
	\begin{figure}
		\centering
		\includegraphics[page=1,scale=0.7]{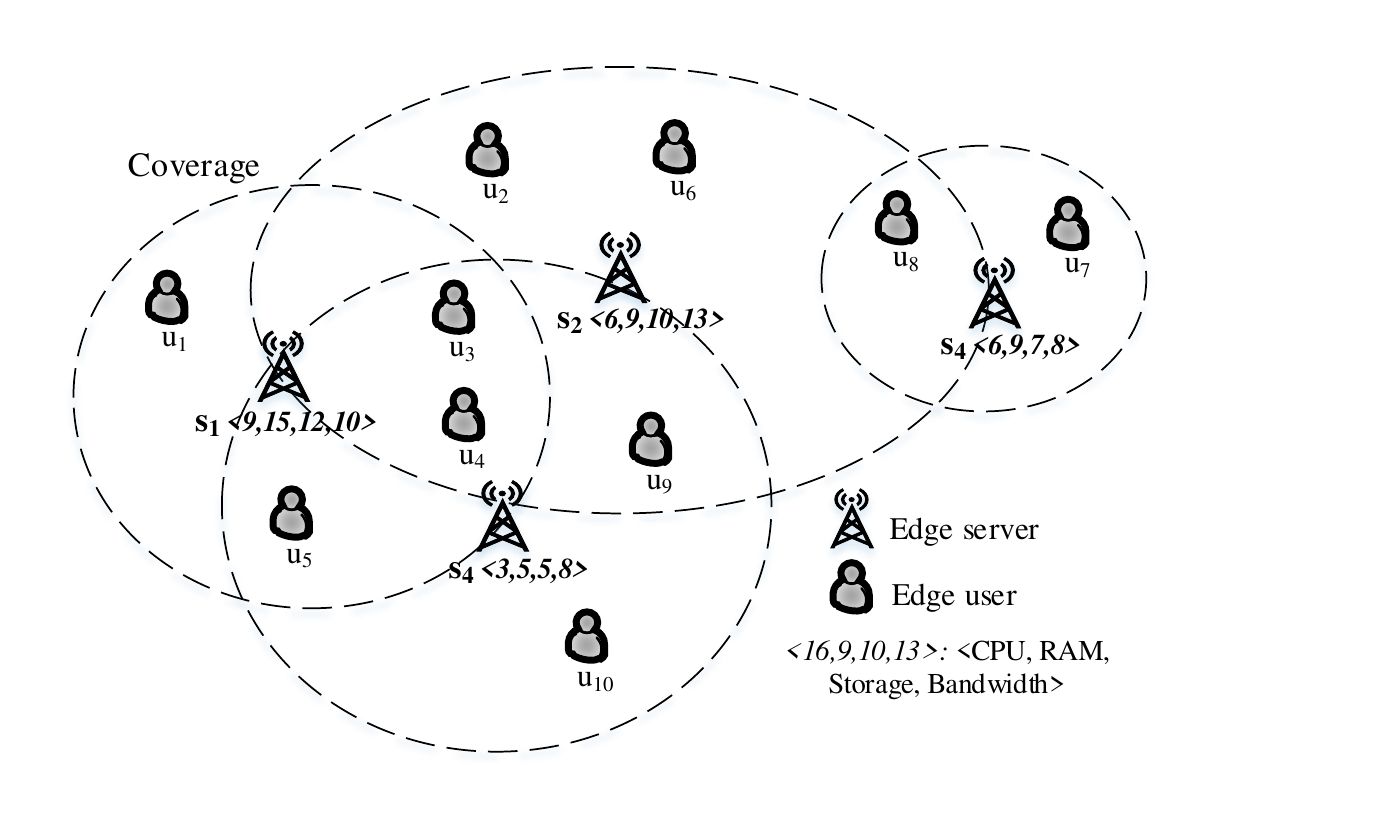}
		\caption{Dynamic QoS EUA example scenario}
		\label{fig:eua_example}
	\end{figure}
	
	Using the game streaming example in Sect. \ref{sec:introduction}, let us consider a simple scenario shown in Fig. \ref{fig:eua_example}. There are ten players $ u_1,...,u_{10} $, and four edge server $ s_1,...,s_4 $. Each edge server has a particular amount of different types of available resources ready to fulfill users' requests. A server's resource capacity or player's resource demand are denoted as a vector $ \langle CPU, RAM, storage, bandwidth \rangle $. The game vendor can allocate its users to nearby edge servers and assign a QoS level to each of them. In this example, there are three QoS levels for the game vendor to choose from, namely $ W_1, W_2 $ and $ W_3 $ (Fig. \ref{fig:QoE_QoS}), which consume $ \langle 1,2,1,2 \rangle $, $ \langle 2,3,3,4 \rangle $, and $ \langle 5,7,6,6 \rangle $ units of $ \langle CPU, RAM, storage, bandwidth \rangle $, respectively. Players' corresponding QoE, measured based on Eq. \ref{eq:qos_model}, are $ 1.6, 4.09 $, and $ 4.99 $, respectively. If the server's available resources are not limited then all players will be able to enjoy the highest QoS level. However, a typical edge server has relatively limited resources so not everyone will be assigned $ W_3 $. The game provider needs to find a player - server - QoS allocation so that the overall user satisfaction, i.e. QoE, is maximized.
	
	Let us assume server $ s_2 $ has already reached its maximum capacity and cannot serve anymore players. As a result, player $ u_8 $ needs to be allocated to server $ s_4 $ along with player $ u_7 $. If player $ u_8 $ is assigned the highest QoS level $ W_3 $, the remaining resources on server $ s_4 $ will suffice to serve player $ u_7 $ with QoS level $ W_1 $. The resulting total QoE of those two players is $ 1.6 + 4.99 = 6.59 $. However, we can see that the released resources from the downgrade from $ W_3 $ to $ W_2 $ allows an upgrade from $ W_1 $ to $ W_2 $. If players $ u_7 $ and $ u_8 $ both receive QoS level $ W_1 $, players' overall QoE is $ 4.09 + 4.09 = 8.18 $, greater than the previous solution.
	
	The scale of the dynamic QoS EUA problem in the real-world scenarios can of course be significantly larger than this example. Therefore, it is not always possible to find an optimal solution in a timely manner, hence the need for an efficient yet effective approach for finding a near-optimal solution to this problem efficiently.
	\newline	
	
	\section{Problem Formulation}\label{sec:problem_formulation}
	\subsection{Problem Definition}\label{sec:problem_definition}
	This section defines the dynamic QoS EUA problem. Table \ref{table:notations} summarizes the notations and definitions used in this paper. Given a finite set of $ m $ edge servers $ \mathcal{S} = \{s_{1},s_{2},...,s_{m}\} $, and $ n $ users $ \mathcal{U} = \{u_{1},u_{2},...,u_{n}\} $ in a particular area, we aim to allocate users to edge servers so that the total user satisfaction, i.e. QoE, is maximized. In the EUA problem, every user covered by edge servers must be allocated to an edge server unless all the servers accessible for the user have reached their maximum resource capacity. If a user cannot be allocated to any edge servers, or is not positioned within the coverage of any edge servers, they will be directly connected to the app vendor's central cloud server.
	
	\begin{table}[H]
		\caption{Key Notations}
		\label{table:notations}
		\begin{tabular}{l|p{8.65cm}}
			\hline
			Notation & Description \\
			\hline
			$ \mathcal{S} = \{s_{1},s_{2},...,s_{m}\} $ & finite set of edge server $ s_{j} $, where $ j = {1,2,...,m} $ \\
			\hline
			\vtop{\hbox{\strut $ \mathcal{D} = \{CPU, RAM,$}\hbox{\strut $ storage, bandwidth\} $}} & a set of computing resource dimension\\
			\hline
			$ c_{j} = \langle c_{j}^1, c_{j}^2, ..., c_{j}^d  \rangle$ & $ d- $dimensional vector with each dimension $ c_{j}^k $ being a resource type, such as CPU or storage, representing the available resources of an edge server $ s_{j} $, $ k \in \mathcal{D} $ \\
			\hline
			$ \mathcal{U} = \{u_{1},u_{2},...,u_{n}\} $ & finite set of user $ u_{i} $, where $ i = {1,2,...,n} $ \\
			\hline
			$ \mathcal{W} = \{W_1,W_2,...,W_q\} $ & a set of predefined resource level $ W_l $, where $ l = 1,2,...,q$. A higher resource level requires more resource than a lower one $ W_l < W_{l+1}.$ We will also refer to a resource level as a QoS level.\\
			\hline
			$ w_{i} = \langle w_{i}^1,w_{i}^2,...,w_{i}^d  \rangle $ & $ d- $dimensional vector representing the resource amount demanded by user $ u_{i} $. Each vector component $ w_{i}^k $ is a resource type, $ k \in \mathcal{D} $. Each user can be assigned a resource level $w_i \in W $ \\
			\hline
			$ \mathcal{U}(s_{j}) $ & set of users allocated to server $ s_{j} $, $ \mathcal{U}(s_{j}) \subseteq \mathcal{U} $ \\
			\hline
			$ \mathcal{S}(u_{i}) $ & set of user $ u_{i} $'s candidate servers -- edge servers that cover user $ u_{i} $, $ \mathcal{S}(u_{i}) \subseteq \mathcal{S} $ \\
			\hline
			$ s_{u_i} $ & edge server assigned to serve user $ u_i $, $ s_{u_i} \in \mathcal{S}$ \\
			\hline
			$ cov(s_{j}) $ & coverage radius of server $ s_{j} $ \\
			\hline
		\end{tabular}
	\end{table}
	
	A user $ u_i $ can only be allocated to an edge server $ s_j $ if they are located within $ s_j $'s coverage area $ cov(s_j) $. We denote $ \mathcal{S}_{u_i} $ as the set of all user $ u_i $'s candidate edge servers -- those that cover user $ u_i $. Take Fig. \ref{fig:eua_example} for example, users $ u_3 $ and $ u_4 $ can be served by servers $ s_1, s_2 $, or $ s_3 $. Server $ s_1 $ can serve users $ u_1, u_3, u_4 $, and $ u_5 $ as long as it has adequate resources.
	
	\begin{equation}\label{eq:edge_coverage}
	u_i \in cov(s_j), \forall u_i \in \mathcal{U}; \forall s_j \in \mathcal{S}
	\end{equation}
	
	If a user $ u_i $ is allocated to an edge server, they will be assigned a specific amount of computing resources $ w_i = (w_i^d) $, where each dimension $ d \in \mathcal{D} $ represents a type of resource, e.g. CPU, RAM, storage, or bandwidth. $ w_i $ is selected from a predetermined set $ \mathcal{W} $ of $ q $ resource levels, ranging from low to high. Each of those resource levels corresponds to a QoS level. The total resources assigned to all users allocated to an edge server must not exceed the available resources on that edge server. The available computing resources on an edge server $ s_j, s_j \in \mathcal{S} $ are denoted as $ c_j = (c_j^d), d \in \mathcal{D} $. In Fig. \ref{fig:eua_example}, users $ u_1, u_3, u_4 $, and $ u_5 $ cannot all receive QoS level $ W_3 $ on server $ s_1 $ because the total required resources would be $ \langle 20,28,24,24 \rangle $, exceeding server $ s_1 $'s available resources $ \langle 9,15,12,10 \rangle $.
	
	\begin{equation} \label{eq:total_capacity}
	\sum\limits_{u_{i} \in \mathcal{U}(s_{j})} w_{i} \leq c_{j}, \quad \forall s_{j} \in \mathcal{S}
	\end{equation}
	
	Each user $ u_i $'s assigned resource $ w_i $ corresponds to a QoS level that results in a different QoE level. As stated in \cite{fiedler2010generic, alreshoodi2013survey, hossfeld2011quantification}, QoS is non-linearly correlated with QoE. When the QoS reaches a specific level, a user's QoE improves very trivially regardless of a noticeable increase in the QoS. For example, in the model in Fig. \ref{fig:QoE_QoS}, the QoE gained from the $ W_2 - W_3 $ upgrade is nearly 1. In the meantime, the QoE gained from the $ W_1 - W_2 $ upgrade is approximately 3 at the cost of a little extra resource.	Several works model the correlation between QoE and QoS using the sigmoid function \cite{hemmati2017qoe, shenker1995fundamental, hande2007distributed}. In this research, we use a logistic function (Equation \ref{eq:qos_model}), a generalized version of the sigmoid function, to model the QoS - QoE correlation. This gives us more control over the QoE model, including QoE growth rate, making the model more generalizable to different domains.
	
	\begin{equation}\label{eq:qos_model}
	E_i = \dfrac{L}{1 + e^{-\alpha(x_i- \beta)}}
	\end{equation}
	where $ L $ is the maximum value of QoE, $ \beta $ controls where the QoE growth should be, or the mid-point of the QoE function, $ \alpha $ controls the growth rate of the QoE level (how steep the change from the minimum to maximum QoE level is), $ E_i $ represents the QoE level given user $ u_i $'s QoS level $ w_i $, and $ x_i = \dfrac{\sum_{k \in \mathcal{D}}w_i^k}{|\mathcal{D}|} $. We let $ E_i = 0 $ if user $ u_i $ is unallocated. 
	
	Our objective is to find a user-server assignment $ \{ u_1,...,u_n \} \longrightarrow \{ s_1,...,s_m \} $ with their individual QoS levels $ \{ w_1,...,w_n \} $ in order to maximize the overall QoE of all users:
	
	\begin{equation}\label{eq:max_qoe}
	maximize \hspace{5pt} \sum_{i=1}^{n} E_i
	\end{equation}
	
	\subsection{Problem Hardness}\label{sec:problem_hardness}
	We can prove that the dynamic QoS EUA problem defined above is $\mathcal{NP}$-hard by proving that its associated decision version is $\mathcal{NP}$-complete. The decision version of dynamic QoS EUA is defined as follows:
	
	Given a set of demand workload $ \mathcal{L} = \{ w_1,w_2,...,w_n \} $ and a set of server resource capacity $ \mathcal{C} = \{c_1,c_2,...,c_m\} $; for each positive number $ Q $ determine whether there exists a partition of $ \mathcal{L}' \subseteq \mathcal{L} $ into $ \mathcal{C}' \subseteq \mathcal{C} $ with aggregate QoE greater than $ Q $, such that each subset of $ \mathcal{L}' $ sums to at most $ c_j, \forall c_j \in \mathcal{C}' $, and the constraint (\ref{eq:edge_coverage}) is satisfied. By repeatedly answering the decision problem, with all feasible combination of $ w_i \in \mathcal{W}, \forall i \in \{1,...,n\} $, it is possible to find the allocation that produces the maximum overall QoE.
	
	\begin{theorem}
		The dynamic QoS EUA problem is $\mathcal{NP}$.
	\end{theorem}
	\begin{proof}
		Given a solution with $ m $ servers and $ n $ users, we can easily verify its validity in polynomial time $ \mathcal{O}(mn) $ -- ensuring each user is allocated to at most one server, and each server meets the condition of having its users' total workload less or equal than its available resource. Dynamic QoS EUA is thus in $\mathcal{NP}$ class.
	\end{proof}
	
	\begin{theorem}
		\textsc{Partition} $ \le_p $ dynamic QoS EUA. Therefore, dynamic QoS EUA is $\mathcal{NP}$-hard.
	\end{theorem}
	
	\begin{proof} 
		We can prove that the dynamic QoS EUA problem is $\mathcal{NP}$-hard by reducing the \textsc{Partition} problem, which is $\mathcal{NP}$-complete \cite{garey2002computers}, to a specialization of the dynamic QoS EUA decision problem.
		\begin{definition}
			(\textsc{Partition}) Given a finite sequence of non-negative integers $ \mathcal{X} = (x_1, x_2,..., x_n) $, determine whether there exists a subset $ \mathcal{S} \subseteq \{1,...,n\} $ such that $ \sum_{i \in \mathcal{S}}x_i = \sum_{j \notin \mathcal{S}}x_j $.
		\end{definition}
		Each user $ u_i $ can be either unallocated to any edge server, or allocated to an edge server with an assigned QoS level $ w_i \in \mathcal{W} $. For any instance $ \mathcal{X} = (x_1, x_2,..., x_n) $ of \textsc{Partition}, construct the following instance of the dynamic QoS problem: there are $ n $ users, where each user $ u_i $ has two 2-dimensional QoS level options, $ \langle x_i,0 \rangle $ and $ \langle 0,x_i \rangle $; and a number of identical servers whose size is $ \langle C,C \rangle $, where $ C = \dfrac{\sum_{i=1}^{n}x_i}{2} $. Assume that all users can be served by any of those servers. Note that $ \langle x_i,0 \rangle \equiv \langle 0,x_i \rangle \equiv w_i $. Clearly, there is a solution to dynamic QoS EUA that allocates $ n $ users to two servers \textit{if and only if} there is a solution to the \textsc{Partition} problem. Because this special case is $\mathcal{NP}$-hard, and being $\mathcal{NP}$, the general decision problem of dynamic QoS EUA is thus $\mathcal{NP}$-complete. Since the optimization problem is at least as hard as the decision problem, the dynamic QoS EUA problem is $\mathcal{NP}$-hard, which completes the proof.
		\newline
	\end{proof}
	
	\section{Our Approach}\label{sec:approaches}
	We first formulate the dynamic QoS EUA problem as an integer linear programming (ILP) problem to find its optimal solutions. After that, we propose a heuristic approach to efficiently solve the problem in large-scale scenarios.
	
	\subsection{Integer Linear Programming Model}\label{sec:approaches_ilp}
	From the app vendor's perspective, the optimal solution to the dynamic QoS problem must achieve the greatest QoE over all users while satisfying a number of constraints. The ILP model of the dynamic QoS problem can be formulated as follows:
	
	\begin{flalign}
	\text{maximize }   	&	\sum\limits_{i=1}^{n}\sum\limits_{j=1}^{m}\sum\limits_{l=1}^{q} E_l x_{ijl}	&&	\label{eq:ilp_max_qoe}\\
	\text{subject to: } &	x_{ijl} = 0	&&	\forall l \in \{1,...,q\}, \forall i,j \in \{i,j|u_i \notin cov(s_j)\} 		\label{eq:ilp_proximity_constraint}\\
	&	\sum\limits_{i=1}^{n} \sum\limits_{l=1}^{q} W_{l}^{k}x_{ijl} \leq c_{j}^{k}	&&	
	\forall j \in \{1,...,m\},\forall k \in \{1,...,d\}	\label{eq:ilp_resource_constraint}\\
	&	\sum\limits_{j=1}^{m} \sum\limits_{l=1}^{q} x_{ijl} \leq	1	&&	\forall i \in \{1,...,n\}	\label{eq:ilp_one_option_constraint}\\
	&	x_{ijl} \in \{0,1\}	&&	\forall i \in \{1,...,n\}, \forall j \in \{1,...,m\}, \forall l \in \{1,...,q\} \notag\label{eq:ilp_domain_constraint}
	\end{flalign}
	
	$ x_{ijl} $ is the binary indicator variable such that,
	\begin{equation}
	x_{ijl}=
	\begin{cases}
	1, & \text{if user $ u_i $ is allocated to server $ s_j $ with QoS level $ W_l $} \\
	0, & \text{otherwise.}
	\end{cases}
	\end{equation}
	
	The objective (\ref{eq:ilp_max_qoe}) maximizes the total QoE of all allocated users. In (\ref{eq:ilp_max_qoe}), the QoE level $ E_l $ can be pre-calculated based on the predefined set $ \mathcal{W} $ of QoS levels $ W_l, \forall l \in \{1,...,q\}$. Constraint (\ref{eq:ilp_proximity_constraint}) enforces the \textit{proximity constraints}. Users not located within a server's coverage area will not be allocated to that server. A user may be located within the overlapping coverage area of multiple edge servers. \textit{Resource constraint} (\ref{eq:ilp_resource_constraint}) makes sure that the aggregate resource demands of all users allocated to an edge server must not exceed the remaining resources of that server. Constraint family (\ref{eq:ilp_one_option_constraint}) ensures that every user is allocated to at most one edge server with one QoS level. In other words, a user can only be allocated to either an edge server or the app vendor's cloud server.
	
	By solving this ILP problem with an Integer Programming solver, e.g. IBM ILOG CPLEX\footnote{www.ibm.com/analytics/cplex-optimizer/}, or Gurobi\footnote{www.gurobi.com/}, an optimal solution to the dynamic QoS EUA problem can be found.
	
	\subsection{Heuristic Approach}\label{sec:approaches_greedy}
	However, due to the exponential complexity of the problem, computing an optimal solution will be extremely inefficient for large-scale scenarios. This is demonstrated in our experimental results presented in Sect. \ref{sec:evaluation}. Approximate methods have been proven to be a prevalent technique when dealing with this type of intractable problems. In this section, we propose an effective and efficient heuristic approach for finding sub-optimal solutions to the dynamic QoS problem.
	
	\begin{algorithm}
		\caption{\textsc{Greedy}}\label{alg:greedy}
		\begin{algorithmic}[1]
			\Procedure{AllocateEdgeUsers}{$ \mathcal{S},\mathcal{U} $}
			\For{each $ u_i \in \mathcal{U}$}
			\State $ \mathcal{S}_{u_i} \gets \{s_j \in \mathcal{S} | u_i \in cov(s_j)\} $;
			\If {$ \mathcal{S}_{u_i} \neq \emptyset $}
			\State $ s_{u_i} \gets \argmax_{s_j \in \{0\}\cup\mathcal{S}_{u_i}}\{s_j : c_j \ge W_1\}  $;
			\State $ w_i \gets \argmax_{W_l \in \{0\}\cup\mathcal{W}}\{W_l : W_l \le c_j\}  $;
			\EndIf 
			\EndFor 
			\EndProcedure
		\end{algorithmic}
	\end{algorithm}
	
	The heuristic approach allocates every user $ u_i \in \mathcal{U} $ one by one (line 2). For each user $ u_i $, we obtain the set $ \mathcal{S}_{u_i} $ of all candidate edge servers that cover that user (line 3). If the set $ \mathcal{S}_{u_i} $ is not empty, or user $u_i$ is covered by one or more edge servers, user $u_i$ will then be allocated to the server that has the most remaining resources among all candidate servers (line 5) so that the server will be most likely to have enough resources to accommodate other users. In the meantime, user $ u_i $ is assigned the highest QoS level that can be accommodated by the selected edge server (line 6). 
	
	The running time of this greedy heuristic consists of: (1) iterating through all $ n $ users, which costs $ \mathcal{O}(n) $, and (2) sorting a maximum of $ m $ candidate edge servers for each user, which costs $ \mathcal{O}(m\log m) $, to obtain the server that has the most remaining resources. Thus, the overall time complexity of this heuristic approach is $ \mathcal{O}(nm \log m) $.
	\newline
	
	
	\section{Experimental Evaluation}\label{sec:evaluation}
	In this section, we evaluate the proposed approaches by an experimental study. All the experiments were conducted on a Windows machine equipped with Intel Core i5-7400T processor(4 CPUs, 2.4GHz) and 8GB RAM. The ILP model in Sect. \ref{sec:approaches_ilp} was solved with IBM ILOG CPLEX Optimizer.
	
	\subsection{Baseline Approaches}\label{sec:evaluation_baselines}
	Our optimal approach and sub-optimal heuristic approach are compared to two other approaches, namely a random baseline, and a state-of-the-art approach for solving the EUA problem:
	\begin{itemize}
		\item \emph{Random}: Each user is allocated to a random edge server as long as that server has sufficient remaining resources to accommodate this user and has this user within its coverage area. The QoS level to be assigned to this user is randomly determined based on the server's remaining resources. For example, if the maximum QoS level the server can achieve is $ W_2 $, the user will be randomly assigned either $ W_1 $ or $ W_2 $.
		\item \emph{VSVBP}: \cite{lai2018optimal} models the EUA problem as a variable sized vector bin packing (VSVBP) problem and proposes an approach that maximizes the number of allocated users while minimizing the number of edge servers needs to be used. Since VSVBP does not consider dynamic QoS, we randomly preset users' QoS levels, i.e., resource demands.
	\end{itemize}
	
	\subsection{Experiment Settings}\label{sec:evaluation_expsettings}
	Our experiments were conducted on the widely-used EUA dataset \cite{lai2018optimal}, which includes data of base stations and end-users within the Melbourne central business district area in Australia. In order to simulate different dynamic QoS EUA scenarios, we vary the following three parameters:
	\begin{itemize}
		\item Number of end-users: We randomly select $ 100, 200,..., 1,000 $ users. Each experiment is repeated 100 times to obtain 100 different user distributions so that extreme cases, such as overly sparse or dense distributions, are neutralized.
		\item Number of edge servers: Say the users selected above are covered by $ m $ servers, we then assume $ 10\%, 20\%, ..., 100\% $ of those $ m $ servers are available to accommodate those users.
		\item Server's available resources: The server's available computing resources is generated following a normal distribution $\mathcal{N}(\mu,\,\sigma^{2})$, where $ \sigma = 1 $ and the average resource capacity of each server $ \mu = 5, 10, 15,... 50 $ in each dimension $ d \in \mathcal{D} $.
	\end{itemize}
	Table \ref{tab:exp_settings} summarizes the settings of our three sets of experiments. The possible QoS level, for each user is preset to $ \mathcal{W} = \{\langle 1,2,1,2 \rangle, \langle 2,3,3,4 \rangle, \\ \langle 5,7,6,6 \rangle \} $. For the QoE model, we set $ L = 5, \alpha = 1.5 $, and $ \beta = 2 $. We employ two metrics to evaluate our approaches: (1) overall QoE achieved over all users for effectiveness evaluation, and (2) execution time (CPU time) for efficiency evaluation.
	\begin{table} 
		\caption{Experiment Settings}
		\label{tab:exp_settings}
		\setlength{\tabcolsep}{8pt}
		\begin{tabularx}{\textwidth}{l|l|l|l}
			\hline
			& Number of users & Number of servers & Server's available resources \\
			\hline
			Set \#1 & $ 100,200,...,1000 $ & $ 70\% $ & $ 35 $ \\
			\hline
			Set \#2 & $ 500 $ & $ 10\%, 20\%,..., 100\% $ & $ 35 $ \\
			\hline
			Set \#3 & $ 500 $ & $ 70\% $ & $ 5,10,15,...,50 $ \\
			\hline
		\end{tabularx}
	\end{table}
	
	\subsection{Experimental Results and Discussion}\label{sec:evaluation_results}
	Figures \ref{fig:results_set1}, \ref{fig:results_set2}, and \ref{fig:results_set3} depict the experimental results of three experiment sets 1, 2, and 3, respectively.
	
	\textit{1) Effectiveness:} Figures \ref{fig:results_set1}, \ref{fig:results_set2}, and \ref{fig:results_set3}(a) demonstrate the effectiveness of all approaches in experiment sets 1, 2, and 3, measured by the overall QoE of all users in the experiment. In general, Optimal, being the optimal approach, obviously outperforms other approaches across all experiment sets and parameters. The performance of Heuristic largely depends on the computing resource availability, which will be analyzed in the following section.
	
	\begin{figure*}[!t]
		\centering
		\subfloat[Total QoE]{\includegraphics[width=2.3in]{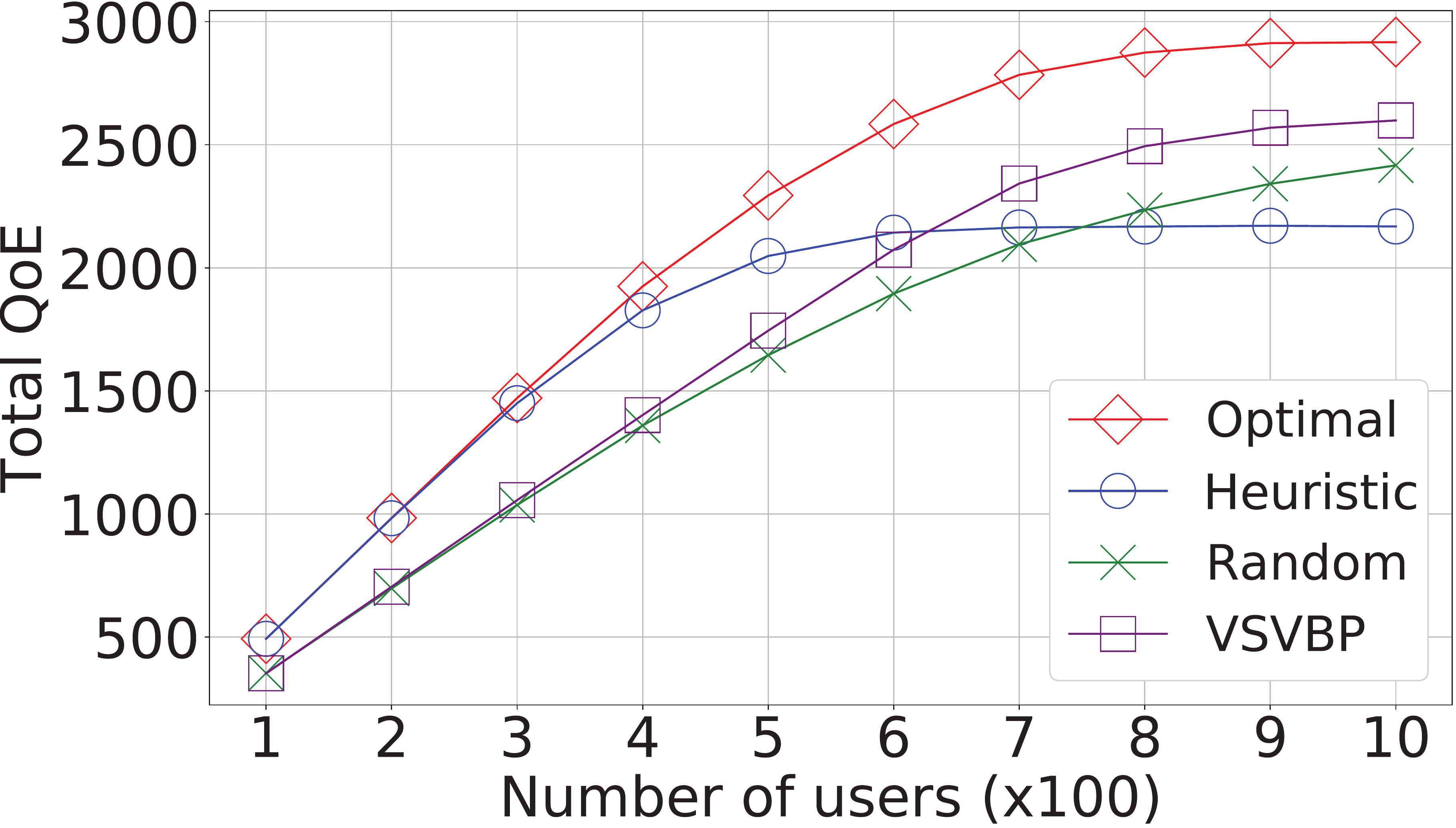}
			\label{fig:total_qoe_set1}}
		\subfloat[Elapsed CPU time]{\includegraphics[width=2.3in]{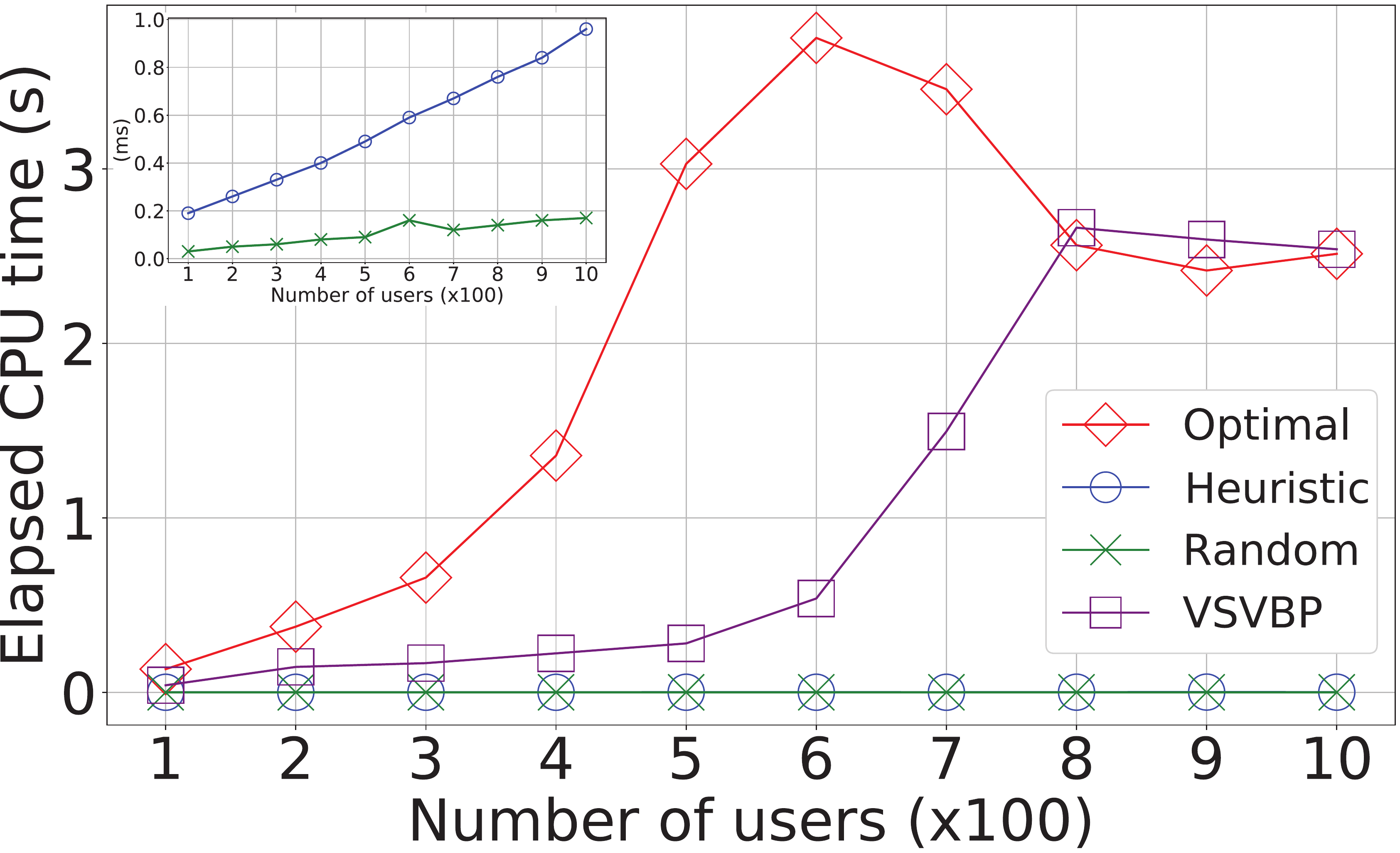}
			\label{fig:cpu_time_set1}}
		\caption{Experiment set \#1 results}
		\label{fig:results_set1}
	\end{figure*}
	
	\begin{figure*}[!t]
		\centering
		\subfloat[Total QoE]{\includegraphics[width=2.3in]{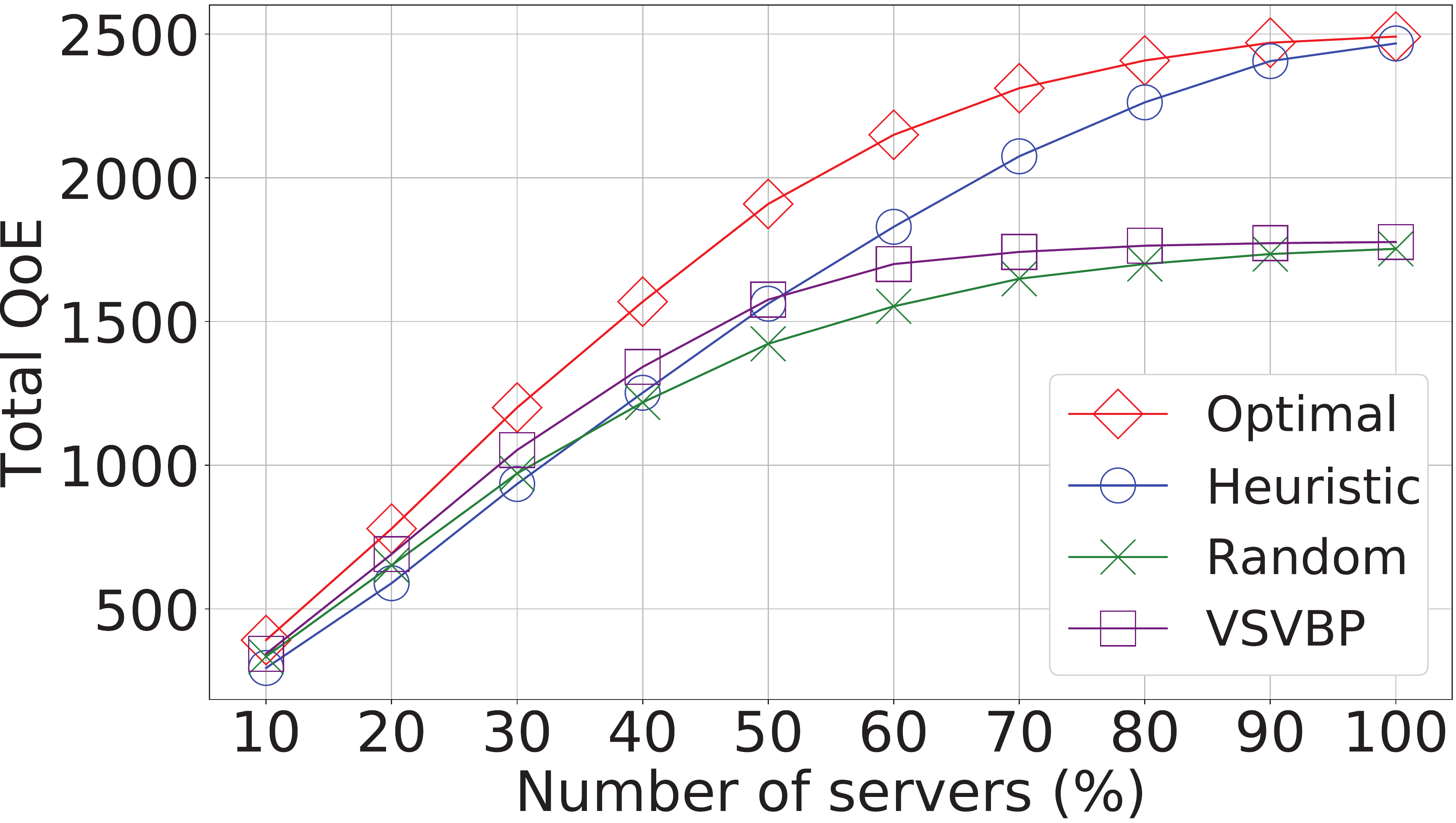}
			\label{fig:total_qoe_set2}}
		\subfloat[Elapsed CPU time]{\includegraphics[width=2.3in]{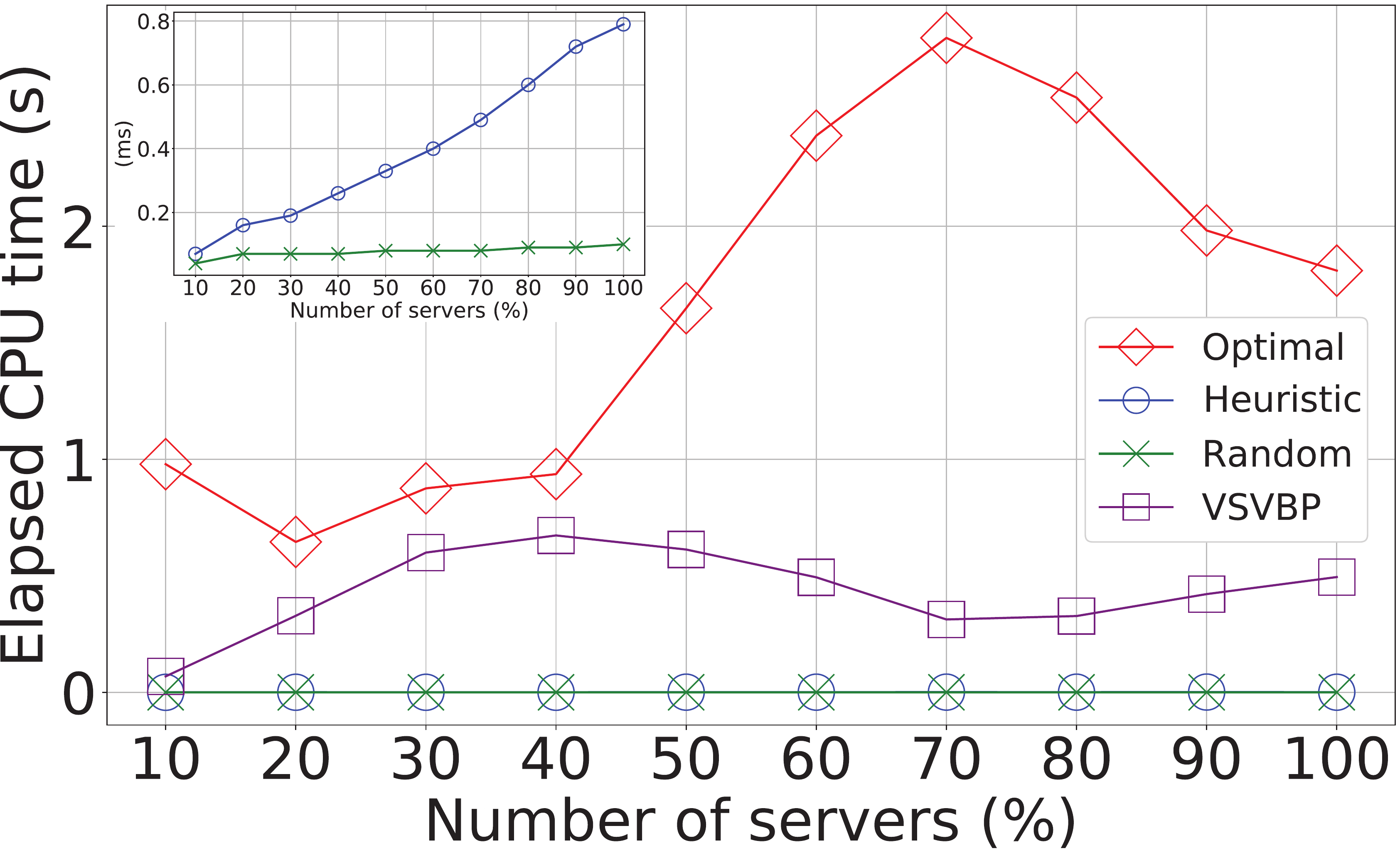}
			\label{fig:cpu_time_set2}}
		\caption{Experiment set \#2 results}
		\label{fig:results_set2}
	\end{figure*}
	
	In experiment set 1 (Fig. \ref{fig:results_set1}(a)), we vary the number of users starting from 100 and ending at 1,000 in steps in 100 users. From 100 to 600 users, Heuristic results in higher total QoE than Random and VSVBP. Especially in the first three steps (100, 200, and 300 users), Heuristic achieves a QoE almost as high as Optimal. This occurs in those scenarios because the available resource is redundant and therefore almost all users receive the highest QoS level. However, as the number of users continues to increase while the amount of available resources is fixed, the computing resource for each user becomes more scarce, making Heuristic no longer suitable in these situations. In fact, from 700 users onwards, Heuristic starts being outperformed by Random and VSVBP. Due to being a greedy heuristic, Heuristic always tries to exhaust the edge servers' resources by allocating the highest possible QoS level to users, which is not an effective use of resource. For example, one user can achieve a QoE of $ 4.99 $ if assigned the highest QoS level $ W_3 $, which consumes a resource amount of $ \langle 5,7,6,6 \rangle $. That resource suffices to serve two users with QoS levels $ W_1 $ and $ W_2 $, resulting in an overall QoE of $ 1.6 + 4.09 = 5.69 > 4.99 $. Since a user's QoS level is randomly assigned by Random and VSVBP, these two methods are able to user resource more effectively than Heuristic in those specific scenarios.
	
	A similar trend can be observed in experiment sets 2 and 3. In resource-scarce situations, i.e. number of servers ranging from 10\% - 40\% (Fig. \ref{fig:results_set2}(a)), and server's available resources ranging from 5 - 25 (Fig. \ref{fig:results_set3}(a)), Heuristic shows a nearly similar performance to Random and VSVBP (slightly worse in a few cases) for the same reason discussed previously. In those situations, the performance difference between Heuristic and Random/VSVBP is not as significant as seen in experiment set 1 (Fig. \ref{fig:results_set1}(a)). Nevertheless, the difference might be greater if the resources are more limited, e.g. 1,000 users in both experiment sets 2 and 3, an average server resource capacity of 20 in set 2, and 50\% number of servers in set 3.
	
	\begin{figure*}[!t]
		\centering
		\subfloat[Total QoE]{\includegraphics[width=2.3in]{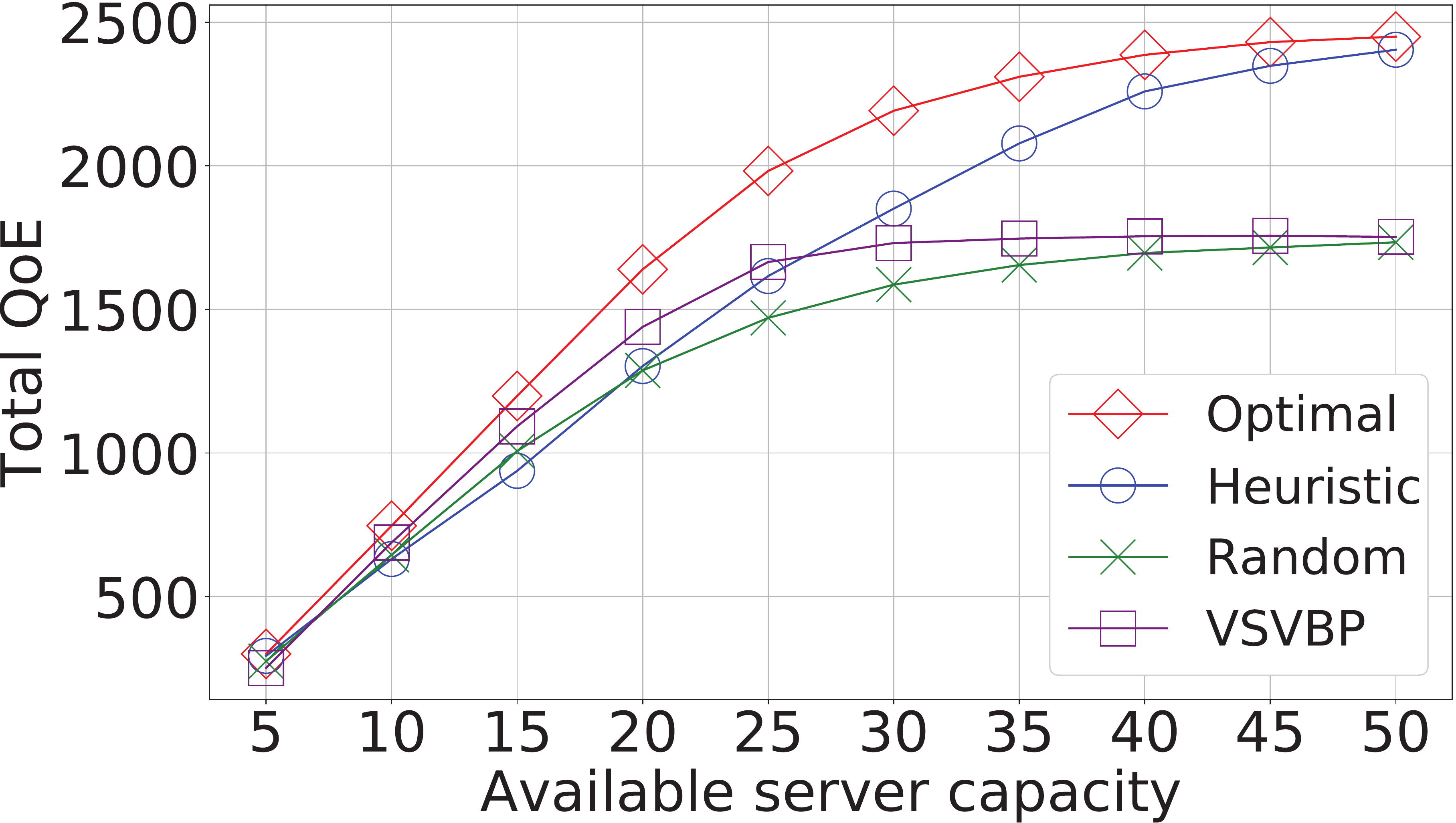}
			\label{fig:total_qoe_set3}}
		\subfloat[Elapsed CPU time]{\includegraphics[width=2.3in]{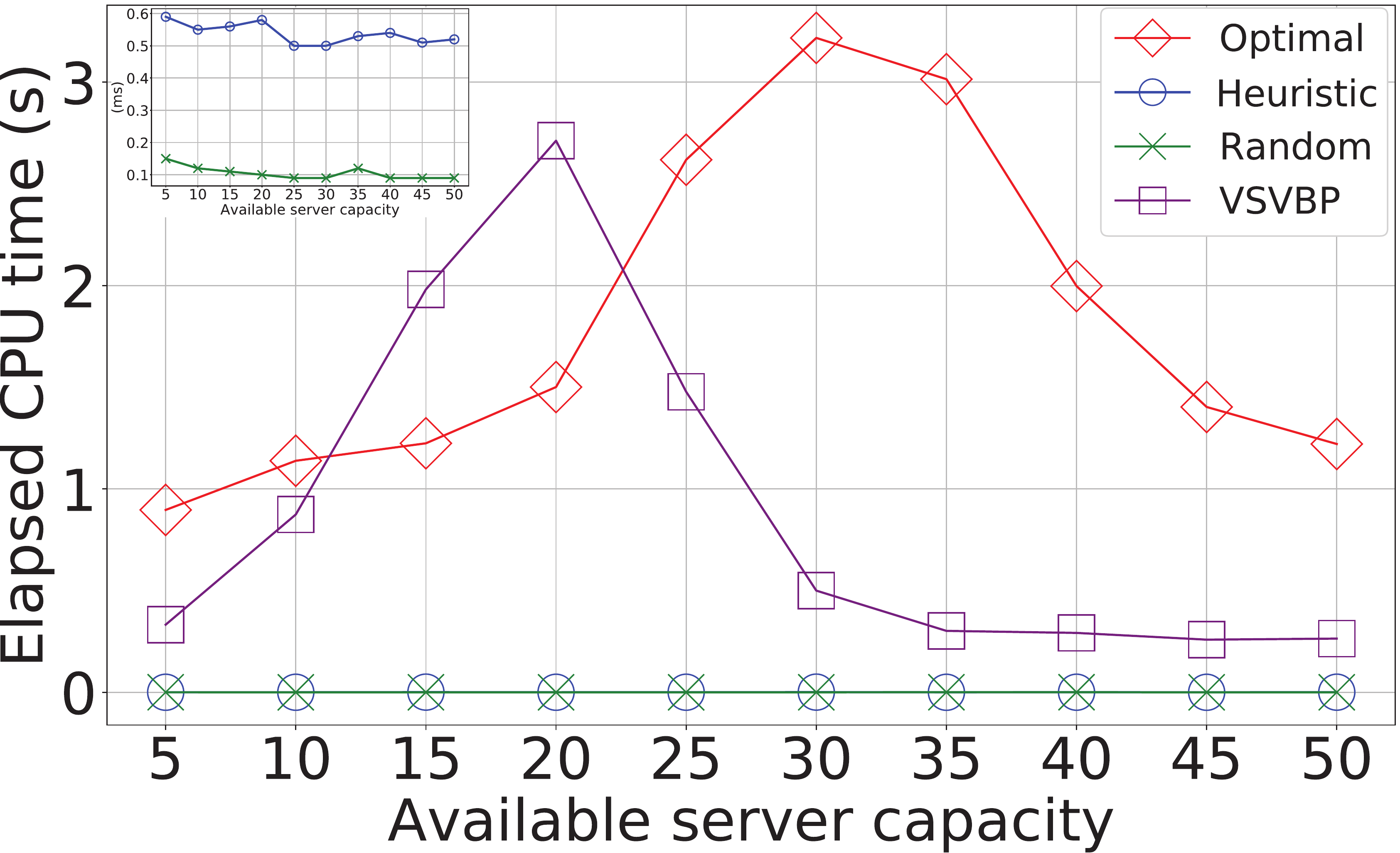}
			\label{fig:cpu_time_set3}}
		\caption{Experiment set \#3 results}
		\label{fig:results_set3}
	\end{figure*}
	
	As discussed above, while being suitable for resource-redundant scenarios, Heuristic has not been proven to be superior when computing resources are limited. This calls for a more effective approach to solve the dynamic QoS problem under resource-scarce circumstances.
	
	\textit{2) Efficiency:} Figures \ref{fig:results_set1}, \ref{fig:results_set2}, and \ref{fig:results_set3}(b) illustrate the efficiency of all approaches in the study, measured by the elapsed CPU time. The execution time of Optimal follows a similar pattern in all three experiment sets. As the experimental parameters increase from the starting point to a point somewhere in the middle -- 600 users in set 1, 70\% number of servers in set 2, and 30 average server resource capacity in set 3 -- the time quickly increases until it reaches a cap of around a hefty 3 seconds due to being $\mathcal{NP}$-hard. The rationale for this is that the complexity of the problem increases as we keep adding up more users, servers, and available resource, generating more possible options and solutions for Optimal to select from. After passing that mid-point, the time gradually decreases at a slower rate then tends to converge. We notice that this convergence is a reflection of the convergence of the total QoE produced by Optimal in each corresponding experiment set. After the experimental parameters passing the point mentioned above, the available resource steadily becomes more redundant so that more users can obtain the highest QoS level without competing with each others, generating less possible options for Optimal, hence running faster. 
	
	In experiment sets 1 and 2, the execution time of Heuristic grows gradually up to just 1 milliseconds. However, it does not grow in experiment set 3 and instead stabilizes around 0.5 - 0.6 milliseconds. This is because the available resource does not impact the complexity of Heuristic, which runs in $ \mathcal{O}(nm \log m) $.
	
	\subsection{Threats to Validity}
	\textit{Threat to construct validity.} The main threat to the construct validity lies in the bias in our experimental design. To minimize the potential bias, we conducted experiments with different changing parameters that would have direct impact on the experimental results, including the number of servers, the number of users, and available resources. The result of each experiment set is the average of 100 executions, each with a different user distribution, to eliminate the bias caused by special cases such as over-dense or over-sparse user distributions.
	
	\textit{Threat to external validity.} A threat to the external validity is the generalizability of our findings in other specific domains. We mitigate this threat by experimenting with different numbers of users and edge servers in the same geographical area to simulate various distributions and density levels of users and edge servers that might be observed in different real-world scenarios.
	
	\textit{Threat to internal validity.} A threat to the internal validity is whether an experimental condition makes a difference or not. To minimize this, we fix the other experimental parameters at a neutral value while changing a parameter. For more sophisticated scenarios where two or more parameters change simultaneously, the results can easily be predicted in general based on the obtained results as we mentioned in Sect. \ref{sec:evaluation_results}.
	
	\textit{Threat to conclusion validity.} The lack of statistical tests is the biggest threat to our conclusion validity. This has been compensated for by comprehensive experiments that cover different scenarios varying in both size and complexity. For each set of experiments, the result is averaged over 100 runs of the experiment.
	\newline
	
	
	\section{Related Work} \label{sec:related_work}
	Cisco \cite{bonomi2012fog} coined the fog computing, or edge computing, paradigm in 2012 to overcome one major drawback of cloud computing -- latency. Edge computing comes with many new unique characteristics, namely location awareness, wide-spread geographical distribution, mobility, substantial number of nodes, predominant role of wireless access, strong presence of streaming and real-time applications, and heterogeneity. Those characteristics allows edge computing to deliver a very broad range of new services and applications at the edge of network, further extending the existing cloud computing architecture.
	
	QoE management and QoE-aware resource allocation have long been a challenge since the cloud computing era and before that \cite{hobfeld2012challenges}. Su et al. \cite{su2016game} propose a game theoretic framework for resource allocation among media cloud, brokers and mobile social users that aims at maximizing user's QoE and media cloud's profit. While having some similarity to our work, e.g. the brokers can be seen as edge servers, there are several fundamental architectural differences. The broker in their work is just a proxy for transferring tasks between mobile users and the cloud, whereas our edge server is where the tasks are processed. In addition, the price for using/hiring the broker/media cloud's resource seems to vary from time to time, broker to broker in their work. We target a scenario where there is no price difference within a single service provider. \cite{he2013cost} investigates the cost - QoE trade-off in virtual machine provisioning problem in a centralized cloud, specific to video streaming domain. QoE is measured by the processing, playback, or downloading rate in those work.
	
	QoE-focused architecture and resource allocation have started gaining attraction in edge computing area as well. \cite{chen2015emc} proposes a novel architecture that integrates resource-intensive computing with mobile application while leveraging mobile cloud computing. Their goal is to provide a new breed of personalized, QoE-aware services. \cite{mahmud2018quality} and \cite{aazam2016mefore} tackle the application placement in edge computing environments. They measure user's QoE based on three levels (low, medium, and high) of access rate, required resources, and processing time. The problem we are addressing, user allocation, can be seen as the step after application placement. \cite{hong2016qoe} focuses on computation offloading scheduling problem in mobile clouds from a networking perspective, where energy and latency must be considered in most cases. They propose a QoE-aware optimal and near-optimal scheduling scheme applied in time-slotted scenarios that takes into account the trade-off between user's mobile energy consumption and latency. 
	
	Apart from the aforementioned literature, there are a number of work on computation offloading or virtual machine placement problem. However, they do not consider QoE, which is important in an edge computing environment where human plays a prominent role. Here, we seek to provide an empirically grounded foundation for the dynamic QoS/QoE edge user allocation problem, forming a solid basis for further developments.
	\newline
	
	
	\section{Conclusion} \label{sec:conclusion}
	App users' quality-of-experience is of great importance for app vendors where user satisfaction is taken seriously. Despite being significant, there is very limited work considering this aspect in edge computing. Therefore, we have identified and formally formulated the dynamic QoS edge user allocation problem with the goal of maximizing users' overall QoE as the first step of tackling the QoE-aware user allocation problem. Having been proven to be $\mathcal{NP}$-hard and also experimentally illustrated, the optimal approach is not efficient once the problem scales up. We therefore proposed a heuristic approach for solving the problem more efficiently. We have also conducted extensive experiments on real-world dataset to evaluate the effectiveness and efficiency of the proposed approaches against a baseline approach and the state of the art. 
	
	Given this foundation of the problem, we have identified a number of possible directions for future work with respect to QoE such as dynamic QoS user allocation in resource-scarce or time-varying situations, user's mobility, service migration, service recommendation, just to name a few. In addition, a finer-grained QoE model with various types of costs or network conditions could be studied next. \newline
	
	\textbf{Acknowledgments.} This research is funded by Australian Research Council Discovery Projects (DP170101932 and DP18010021).
	\newline
	
	\bibliographystyle{splncs04}
	\bibliography{IEEEabrv,ms}
	
\end{document}